\documentclass[runningheads]{llncs}

\usepackage[none]{hyphenat}
\usepackage{setspace}
\usepackage{latexsym,amsmath}
\usepackage{amsfonts}
\usepackage{algorithm}
\usepackage[noend]{algpseudocode}
\usepackage{graphicx}
\usepackage[dvipsnames]{xcolor}
\usepackage{soul}
\usepackage{comment}
\usepackage{multibib}

\newcites{supp}{Appendix References}

\makeatletter
\def\BState{\State\hskip-\ALG@thistlm}
\makeatother

\algdef{SE}[SUBALG]{Indent}{EndIndent}{}{\algorithmicend\ }
\algtext*{Indent}
\algtext*{EndIndent}

\includecomment{appendixonly}
\excludecomment{mainonly}

\newcounter{partcounter}
\newcommand{\myprefix}{}

\newcommand{\mylabel}[1]{\label{\myprefix#1}}
\newcommand{\myref}[1]{\ref{\myprefix#1}}
\newcommand{\myeqref}[1]{\eqref{\myprefix#1}}

\newcommand{\setprefix}[1]{\renewcommand{\myprefix}{#1}}

\makeatletter
\newcommand{\mypartprefix}{%
  \ifnum\value{partcounter}>0 
    \Alph{partcounter}%
  \fi
}

\renewcommand{\thefigure}{\mypartprefix\arabic{figure}}

\newcommand{\mypart}{
  \stepcounter{partcounter}  
  \setcounter{section}{0}
  \setcounter{theorem}{0}
  \setcounter{lemma}{0}
  \setcounter{figure}{0}
  \setcounter{algorithm}{0}
  \setcounter{definition}{0}
}

\numberwithin{equation}{section}

\begin{document}
\title{Linear Search for an Escaping Target with Unknown Speed\thanks{This is the full version of the paper which appeared in the Proceedings of IWOCA 2024: 35th International Workshop on Combinatorial Algorithms, SVLNCS, 1-3 July 2024 Ischia, Italy.}
}

\author{
Jared Coleman\inst{1}
\and
Dmitry Ivanov\inst{2}
\and
Evangelos Kranakis\inst{2}~\inst{5}
\and
Danny Krizanc\inst{3}
\and
Oscar Morales Ponce\inst{4}
}

\institute{
University of Southern California, California, USA \\
\email{jaredcol@usc.edu}
\and
Carleton University, Ottawa, Ontario, Canada \\
\email{dimaivanov@cmail.carleton.ca} \\
\email{kranakis@scs.carleton.ca}
\and
Wesleyan University, Middletown CT, USA \\
\email{dkrizanc@wesleyan.edu}
\and
California State University, Long Beach, USA \\
\email{Oscar.MoralesPonce@csulb.edu}
\and
Research supported in part by NSERC Discovery grant.
}

\authorrunning{J. Coleman et al.}

\maketitle


\begin{abstract}
We consider linear search for an escaping target whose speed and initial position are unknown to the searcher. A searcher (an autonomous mobile agent) is initially placed at the origin of the real line and can move with maximum speed $1$ in either direction along the line. An oblivious mobile target that is moving away from the origin with an unknown constant speed $v<1$ is initially placed by an adversary on the infinite line at distance $d$ from the origin in an unknown direction. We consider two cases, depending on whether $d$ is known or unknown. The main contribution of this paper is to prove a new lower bound and give algorithms leading to new upper bounds for search in these settings. This results in an optimal (up to lower order terms in the exponent) competitive ratio in the case where $d$ is known and improved upper and lower bounds for the case where $d$ is unknown. Our results solve an open problem proposed in [Coleman et al., Proc. OPODIS 2022]. 

\vspace{0.5cm}

\noindent
{\bf Key words and phrases.} Autonomous agent, Competitive ratio, Linear Search, Oblivious Mobile Target, Searcher.

\end{abstract}

\mypart
\setprefix{}
\excludecomment{mainonly}
\includecomment{appendixonly}

\section{Introduction}\mylabel{sec:intro}
Linear search is concerned with a searcher that wants to find a target located on the infinite real line. The {\em searcher} is an autonomous mobile agent, henceforth also referred to as the \textit{robot}, that can move on the line in any direction with max speed $1$. A second agent, henceforth also referred to as the {\em target}, is oblivious (in particular, it cannot change direction), starts at an unknown (to the searcher) location on the line, and moves with constant speed $v<1$ away from the origin. The goal is to design search algorithms that minimize the ratio of the target capture time (i.e., the first time when the searcher and target are co-located on the line) under the algorithm versus the minimum time required by an agent that knows the speed and starting position of the target.
We refer to this ratio as the \textit{competitive ratio}.
For many problems, it is possible to design an algorithm with a finite supremum competitive ratio amongst all problem instances. In our case, however, the competitive ratio grows arbitrarily large for certain instances; so, our goal will be to minimize the growth of the competitive ratio with respect to the target's initial parameters.

One would expect that the resulting competitive ratio of a search algorithm should depend not only on the moving target's initial parameters (speed $v$ and starting distance $d$), but also on whether those parameters are known to the searcher prior to the execution of the algorithm. It would be natural to speculate that the less the searcher knows about the target, the worse the resulting competitive ratio. Our main focus in this paper is to understand how the knowledge available to the searcher about $d$, $v$ affects the competitive ratio of linear search. More specifically, we ask: what is the competitive ratio of search if the target is mobile and its speed is fixed but \textit{unknown} to the searcher? This question was first considered in~\cite{coleman2023line}, where the authors study the resulting competitive ratios of search under the four possible knowledge scenarios of the pair $d, v$. The main contribution of the present paper is to give new algorithms and determine new upper and lower bounds for the competitive ratios when $v$ is unknown to the searcher, and thus
answer one of the main open problems left in~\cite{coleman2023line}.

\subsection{Notation and Terminology}\mylabel{sec:notation}
A robot with maximum speed $1$ starts at the origin $0$ on the real number line. A target starts a distance $d \geq 1$ from the origin in an unknown direction and moves away from the origin with an unknown speed $0\leq v < 1$. The goal is for the robot to catch the target in a manner that is as efficient as possible. With the knowledge that it has, the robot executes an algorithm until it becomes collocated with the target. For simplicity, we define the positive direction of the number line to be that in which the robot's position reaches magnitude $d$ first.

The only feedback that the robot receives during its search is either the instruction to keep searching or the instruction to terminate (once it is collocated with the target).
As such, its motion is predetermined by the initial parameters that are available to it.
We will refer to this predetermined pattern of motion as a ``strategy'' and define it as follows.

\begin{definition} \mylabel{def:strat}
    A strategy $\mathcal{A}$ is a function $\mathcal{A}(t)$ for $t\geq 0$ that denotes the robot's position on the number line at time $t$.
    $\mathcal{A}(t)$ is continuous, does not violate the speed limit of $1$, and $\mathcal{A}(0)=0$.
    A strategy $\mathcal{A}$ is ``successful'' if for any possible target, there exists some $t$ at which the robot is collocated with that target.
\end{definition}

Since one strategy must account for all problem instances,
it is useful to imagine the targets from all problem instances as existing simultaneously.
In such a paradigm, the goal of a strategy is to ``catch'' all possible targets (whereas only one is caught in reality).
The time taken to catch each potential target can be used to evaluate the efficiency of the strategy.
We will use the word ``reach'' to mean becoming collocated with a target whereas the word ``catch'' will mean reaching a target for the first time.

The fraction $u=\frac{1}{1-v}$ represents the so-called {\em evasiveness} of a target of speed $v$.
The evasiveness of a target is proportional to the minimum amount of time that it takes for the robot to close a gap between itself and the target. Note that $v = 1 - \frac1u$ and by the bounds of $v$, we have $1\leq u < \infty$.
From here on out, it is useful to frame the conversation in terms of evasiveness rather than speed.
Every variable denoted by the letter $u$ will represent an evasiveness and will implicitly be associated with a speed variable denoted by the letter $v$ with the same subscript or lack thereof.

The variable $\sigma \in \left\{0,1\right\}$ will be used to represent the positive side of the real line if $\sigma=0$ or the negative side if $\sigma=1$.
We define
$$
    T^\sigma_{\mathcal{A}}(u,d) = \inf \left\{t\geq 0 : \mathcal{A}(t)=(-1)^\sigma\left(d+\left(1-\frac{1}{u}\right)t\right)\right\}
$$
as the time taken by strategy $\mathcal{A}$ to catch $u$ on side $\sigma$ and $T_{opt}(u,d) = \frac{d}{1-v} = ud$ as the time taken by a robot running an offline algorithm to accomplish the same task. The competitive ratio is defined as $$CR^\sigma_{\mathcal{A}}(u,d) = \frac{T^\sigma_{\mathcal{A}}(u,d)}{T_{opt}(u,d)}\mbox{ and }
    CR_{\mathcal{A}}(u,d) = \max\left(CR^0_{\mathcal{A}}(u,d),CR^1_{\mathcal{A}}(u,d)\right) .
$$
In all cases, we will be using $CR_{\mathcal{A}}(u,d)$ to evaluate the performance of search strategies.
When the starting distance $d$ is known to the searcher (as is the case in Section~\myref{sec:no_speed}), we will suppress $d$ in the definitions above and simply write $T^\sigma_{\mathcal{A}}(u)$, $ T_{opt}(u)$, $CR^{\sigma}_{\mathcal{A}}(u)$, or $CR_{\mathcal{A}}(u)$, respectively to simplify the notation.

\subsection{Related work}\mylabel{sec:related}
The study of linear search for a static target was first proposed independently in \cite{beck1964linear,bellman1963optimal} for stochastic and game theoretic models, while deterministic linear search by a single mobile agent was investigated in~\cite{baezayates1993searching,baezayates1995parallel}.
There is extensive literature on linear search for both stochastic and deterministic search models. For additional information, the interested reader could consult the seminal books by Ahlswede and Wegener~\cite{ahlswede1987search} and Alpern and Gal~\cite{alpern03}.

The ultimate goal of linear search is to determine ``precise bounds'' on the competitive ratio of a specific search model. As such, the competitive ratio may play a primary role in determining the computational strength of a search model as well as its relation to other existing models. This is particularly evident when analyzing the impact of communication in various settings including multi-agent (group) search with or without faulty agents, different communication abilities, etc. For example, see the work of~\cite{DBLP:journals/dc/CzyzowiczKKNO19,kupavskii2018lower} for linear group search with crash faults, and~\cite{DBLP:journals/ijfcs/CzyzowiczGKKNOS21,killick2021ACDA,sun2020better} for the case of byzantine faults. Another example, is when two agents can communicate only in the F2F (Face-to-Face) model (see~\cite{chrobak2015group}) as well as \cite{bampas2019linear} which analyzes the competitive ratio of linear search for two robots with different speeds. It should also be noted that there is extensive literature when different cost models are considered, e.g., \cite{demaine2006online} which takes into account the number of turns.

Most of the known research on analyzing and determining the competitive ratio of linear search is concerned with a static target. Alpern and Gal~\cite{alpern03}[p. 134, Equation 8.25] were the first to analyze linear search for catching an escaping oblivious target and proved that the optimal competitive ratio of search for a moving target starting at an unknown starting distance and moving with known (to the searcher) speed $0 \leq v <1$ (where $1$ is the speed of the searcher) away from the origin is exactly $1 + 8\frac{1+v}{(1-v)^2}$. If both starting distance $d$ and speed $v$ of the mobile target are known to the searcher, the competitive ratio is shown in~\cite{coleman2023line} to be $1+\frac2{1-v}$. This last publication was also the first work to consider knowledge/competitive ratio tradeoffs in linear search for a moving target. In the present paper, we consider the general problem of the impact of knowing either $d$ or $v$ on the competitive ratio of search. In particular, we focus on determining the competitive ratio of search for a moving target when its speed is unknown to the searcher and solve one of the main open problems proposed in ~\cite{coleman2023line}.

\subsection{Contributions and Outline}\mylabel{sec:outline}
Recall that, throughout the paper, we assume that the speed $v$ of the mobile target is unknown to the searcher and $u=\frac1{1-v}$ is the evasiveness of $v$. The main contributions are as follows.

In Section~\myref{sec:no_speed}, we assume the starting distance $d$ of the mobile target is known to the searcher. In Subsection~\myref{sec:no_speed:lower}, we prove a lower bound in Theorem~\myref{thm:no_speed:lower}: that no search strategy $\mathcal{A}$ can satisfy $CR_{\mathcal{A}}(u) \in O\left(u^{4-\varepsilon}\right)$, for any constant $\varepsilon>0$ (improving the lower bound of $1+2u$ in~\cite{coleman2023line}). In Subsection~\myref{sec:no_speed:upper}, we define strategy $\mathcal{A}_1$ (see Algorithm~\myref{alg:no_speed:upper:sqrt_for_upbound}) and prove, in Theorem~\myref{thm:nospeed_upbound}, that it obeys $CR_{\mathcal{A}_1}(u) \leq 56.18u^{4-(\log_2 \log_2 u)^{-2}}$ for $u>4$, while $CR_{\mathcal{A}_1}(u)=9$ for $1 \leq u \leq 4$
(improving the previous upper bound of $1+16 u^4 (\log_2 u )^2$ from~\cite{coleman2023line} for all $u>4$).
Our results for Section \myref{sec:no_speed} are tight up to lower order terms in the exponent, thus answering one of the main open problems proposed in~\cite{coleman2023line}.

In Section~\myref{sec:no_knowledge}, we analyze the no-knowledge case whereby the searcher knows neither the speed $v$ nor the starting distance $d$ of the target. We define a strategy $\mathcal{A}_2$ (see Algorithm~\myref{alg:no_knowledge}) and prove, in Theorem~\myref{thm:noknow_upbound_cr2param}, that it obeys $CR_{\mathcal{A}_2}(u,d)\leq 1+\frac{1}{d}\left(56.18(ud)^{4-(\log_2 \log_2 (ud))^{-2}}-1\right)$ for $ud>4$, while $CR_{\mathcal{A}_2}(u,d)\leq 1+\frac 8d$ for $ud \leq 4$ (improving the previous best upper bound: $O \left(\frac1d M^8 \log_2^2 M \log_2 \log_2 M\right)$, where $M = \max\left(u,d\right)$, in~\cite{coleman2023line}). We also note that the lower bound in the case where $d$ is known extends trivially to the case where $d$ is unknown, yielding that no strategy $\mathcal{A}$ could satisfy $CR_{\mathcal{A}}(u,d) \in O\left(u^{4-\varepsilon}\right)$, for any constant $\varepsilon>0$.
We conclude in Section~\myref{sec:conclusion} by summarizing our contributions and proposing some open problems.
\begin{mainonly}
A full version of the paper with all proofs omitted due to space constraints can be found in the full version of the paper~\cite{arxiv_version}.
\end{mainonly}

\section{Unknown Speed and Known Starting Distance}
\mylabel{sec:no_speed}
In this section, we consider search when the speed $0 \leq v <1$ of the target is unknown but the starting distance $d$ is known to the searcher. First, we prove a lower bound in Subsection~\myref{sec:no_speed:lower}, followed by an upper bound in Subsection~\myref{sec:no_speed:upper}.

\subsection{Lower Bound}\mylabel{sec:no_speed:lower}

In Definition~\myref{def:strat}, we defined the general concept of a search strategy. Next, we define a particular class of search strategies: so-called zigzag strategies.
\begin{definition} \mylabel{def:zigzag}
    A zigzag strategy $\mathcal{A}$ is defined in terms of a sequence of evasiveness values: $\left(u_i\right)_{i=0}^{\infty}$. This sequence must satisfy the additional stipulations that $u_{i+2} > u_i$ for all $i\geq 0$ and that $\lim_{i \to \infty} u_i = \infty$. The movement of the robot is divided into rounds with indices $i\geq 0$. On round $i$, the robot travels from the origin to position $(-1)^ix_i$ and back to the origin where $x_i$ is the distance that is just large enough to catch the target with evasiveness $u_i$. For future convenience, let $s_i = \sum_{n=0}^{i}x_n$ for all $i\geq -1$. To be clear, $s_{-1}=0$.
\end{definition}

For any successful strategy, there exists a zigzag strategy that performs just as well or better. This is because it is optimal for the robot to always be travelling to the next new target that it intends to catch at top speed (see~\cite{alpern03}). For this reason, we will only be considering zigzag strategies as they are the best-performing class of strategies. This subsection is dedicated to proving the following theorem.

\begin{theorem} \mylabel{thm:no_speed:lower}
    In the case where initial distance $d$ is known but speed $v$ is unknown, no strategy $\mathcal{A}$ can satisfy $CR_{\mathcal{A}}(u) \in O\left(u^{4-\varepsilon} \right)$, for any constant $\varepsilon > 0$.
\end{theorem}
\begin{proof}
    This proof will show that a contradiction arises from the assumption that $CR_{\mathcal{A}}(u) \in O\left(u^k\right)$ for an arbitrary zigzag strategy $\mathcal{A}$ and for an arbitrary constant $k<4$.
    Part of the proof borrows work done by Beck and Newman \cite{beck1970linearsearch} regarding the cow path problem.
    Here is an outline of the proof:
    \begin{enumerate}
        \item Lemmas \myref{lem:no_speed:lower:cru} through \myref{lem:no_speed:lower:li_uprod} will give a lower bound for $CR_{\mathcal{A}}(u)$ in terms of $u_i$ for various $i$.
        \item This bound will be distilled into a key condition, but only under the assumption that $CR_{\mathcal{A}}(u) \in O\left(u^k\right)$ for $k<4$.
        \item Lemma~\myref{lem:no_speed:lower:k_condition} will show that this condition is impossible to satisfy.
    \end{enumerate}
    Note that Lemma~\myref{lem:si_recurrence} will also be useful in later proofs.

    \begin{lemma} \mylabel{lem:no_speed:lower:cru}
        For some zigzag strategy $\mathcal{A}$, for any round number $i\geq 0$, and for any evasiveness value $u\in (u_i,u_{i+2}]$, $CR_{\mathcal{A}}(u)\geq 1+\frac{2}{d}s_{i+1}$.
    \end{lemma}
    \begin{proof}
        Consider a round $i\geq 0$ with parity $\sigma$.
        During round $i+2$, all $u\in (u_i,u_{i+2}]$ on side $\sigma$ are caught (for the first time).
        At the start of round $i+2$, the robot is at the origin but has spent $2s_{i+1}$ time on previous rounds, giving all targets an additional $2vs_{i+1}$ distance head start.
        The target with evasiveness $u$ is caught in time
        \begin{align*}
            T^{\sigma}_{\mathcal{A}}(u) & =\; 2s_{i+1} + \frac{d+2vs_{i+1}}{1-v}
            =\; \frac{d+2vs_{i+1}+2s_{i+1}(1-v)}{1-v}
            =\; ud+2us_{i+1}.
        \end{align*}
        It is true that
        $
            T^{\sigma}_{\mathcal{A}}(u) \leq \max\left(T^0_{\mathcal{A}}(u), T^1_{\mathcal{A}}(u)\right)
        $.
        Hence, by definition,
        $
            CR_{\mathcal{A}}(u) \geq \frac{T^{\sigma}_{\mathcal{A}}(u)}{T_{opt}(u)}\nonumber\\
            = \frac{ud+2us_{i+1}}{ud}
            = 1+\frac{2}{d}s_{i+1}
        $    .
        This completes the proof of Lemma~\myref{lem:no_speed:lower:cru}. ~\qed
    \end{proof}

    \begin{lemma}\mylabel{lem:si_recurrence}
        For any zigzag strategy $\mathcal{A}$, for all $i\geq 0$, $s_i = u_id+\left(2u_i-1\right)s_{i-1}$.
    \end{lemma}
    \begin{proof}
        At the start of round $i$, the robot is at the origin and must catch evasiveness values up to $u_i$.
        The robot has spent $2s_{i-1}$ time on previous rounds.
        Therefore, the target with evasiveness $u_i$ is a distance $d+2v_is_{i-1}$ away from the robot at the start of round $i$.
        Catching this target will take an amount of time (and distance) equal to
        \begin{align}
            x_i & = u_i\left(d+2v_is_{i-1}\right)
            = u_i\left(d+2\left(1-\frac{1}{u_i}\right)s_{i-1}\right)
            = u_id+(2u_i-2)s_{i-1}. \nonumber
        \end{align}
        By definition,
        $        s_i = x_i+s_{i-1}
            = u_id+(2u_i-2)s_{i-1}+s_{i-1}
            = u_id+(2u_i-1)s_{i-1}.
        $
        This proves Lemma~\myref{lem:si_recurrence}.~\qed
    \end{proof}

    \begin{lemma}\mylabel{lem:no_speed:lower:li_uprod}
        For any zigzag strategy $\mathcal{A}$, for all $i\geq 0$, $s_i \geq d\prod_{n=0}^{i}u_n$.
    \end{lemma}
    \begin{proof}
        Lemma~\myref{lem:si_recurrence} states that $s_i = u_id+\left(2u_i-1\right)s_{i-1}$.
        Since $u_i\geq 1$,
        $$
            s_i = u_id+(2u_i-1)s_{i-1}
            \geq u_id+(2u_i-u_i)s_{i-1}
            = u_id+u_is_{i-1}
            \geq u_is_{i-1}.
        $$
        In summary, for all $i\geq 0$, $s_i \geq u_is_{i-1}$.
        Using $s_0=u_0d$ as the base case, we can unwrap the recurrence to yield $s_i \geq d\prod_{n=0}^{i}u_n$.
        This completes the proof of Lemma~\myref{lem:no_speed:lower:li_uprod}. ~\qed
    \end{proof}

    Now, for the sake of contradiction, assume that there exists a zigzag strategy $\mathcal{A}$ satisfying $CR_{\mathcal{A}}(u)\in O\left(u^k\right)$ for some $k<4$.
    More specifically, assume that there exists some $a>0$ such that for all $u \geq c_1$ for some constant $c_1$,
    \begin{equation}
        CR_{\mathcal{A}}(u) \leq au^k . \mylabel{eq:verbose_bigo}
    \end{equation}
    We now examine a subset of all $u$ values on which this condition must hold.
    For each round numbered $i\geq 0$, let us select an evasiveness $u=\alpha_iu_i$ where
    \begin{equation}
        \alpha_i = \min\left(\frac{u_{i+2}}{u_i},\sqrt[k]{2}\right) . \mylabel{eq:alpha_coeff}
    \end{equation}
    Note that $u_i<u_{i+2}$ by definition and that we need only consider $k>0$.
    With this in mind, one can verify that $\alpha_iu_i$ satisfies the following properties:
    \begin{align}
        u_i < \alpha_iu_i \leq u_{i+2} , \mylabel{eq:alpha_range_cond} \\
        \left(\alpha_iu_i\right)^k \leq 2u_i^k . \mylabel{eq:alpha_scale_cond}
    \end{align}
    Lemma~\myref{lem:no_speed:lower:cru} in conjunction with Equation \myeqref{eq:alpha_range_cond} yields:
    $
        CR_{\mathcal{A}}\left(\alpha_iu_i\right) \geq 1+\frac{2}{d}s_{i+1}
    $.
    Lemma~\myref{lem:no_speed:lower:li_uprod}, Equation \myeqref{eq:verbose_bigo}, and Equation \myeqref{eq:alpha_scale_cond} yield:
    \begin{equation*}
        1+2\prod_{n=0}^{i+1}u_n \;\leq\;
        1+\frac{2}{d}s_{i+1} \;\leq\;
        CR_{\mathcal{A}}\left(\alpha_iu_i\right) \;\leq\;
        a\left(\alpha_iu_i\right)^k \;\leq\;
        2au_i^k .
    \end{equation*}
    Now, let $u_i=2^{w_i}$ where $w_i\geq 0$.
    Note that since $u_i$ goes to infinity by the definition of a zigzag strategy, so does $w_i$. Hence, $2a\cdot 2^{kw_i} \geq 1+2\prod_{n=0}^{i+1}2^{w_n} \geq 2\prod_{n=0}^{i+1}2^{w_n}$.
    Dividing by $2$ and applying $\log_2$ to both sides produces $kw_i+\log_2(a) \geq \sum_{n=0}^{i+1}w_n$.
    We can remove $\log_2(a)$ by swapping $k$ with $k_2=\frac{1}{2}(k+4)$, which is bigger than $k$ but still less than $4$.
    Since $\log_2(a)$ is a constant and $w_i$ goes to infinity, there will necessarily be some round $j$ such that for all $i\geq j$,
    $kw_i+\log_2(a) \leq \frac{1}{2}(k+4)w_i$.
    Alternatively, for all $i\geq j$, $w_i\geq \frac{2\log_2(a)}{4-k}$.
    Let $c_2 = u_j$.
    We now have that for all $i$ satisfying $u_i \geq \max(c_2, c_1)$,
    \begin{equation}
        k_2w_i \geq \sum_{n=0}^{i+1}w_n .\mylabel{eq:no_speed:lower:main_cond}
    \end{equation}
    We will now show that Equation \myeqref{eq:no_speed:lower:main_cond} cannot hold using Lemma~\myref{lem:no_speed:lower:k_condition}, which adapts part of Lemma~2 from \cite{beck1970linearsearch} to better suit the task at hand.
    \begin{lemma} \mylabel{lem:no_speed:lower:k_condition}
        Consider a sequence $\left(z_i\right)_{i=0}^{\infty}$ where $z_i>0$ for all $i\geq 2$ and $z_i\geq 0$ for $i=0,1$.
        For $h<4$, no such sequence is able to satisfy the following condition for all $i$ larger than a sufficiently large $m$:
        \begin{equation}
            hz_i \geq \sum_{n=0}^{i+1}z_n .\mylabel{eq:no_speed:lower:k_cond}
        \end{equation}
    \end{lemma}
    \begin{proof}
        For the sake of contradiction, assume that a sequence $\left(z_i\right)_{i=0}^{\infty}$ can satisfy the above conditions.
        Consider a modified sequence $\left(y_i\right)_{i=0}^{\infty}$ given by
        \begin{equation}
            y_i = 2^{-i}\sum_{n=0}^{i}z_n .\mylabel{eq:modified_y_sequence}
        \end{equation}
        The conditions of the lemma ensure that $y_i>0$ for all $i\geq 2$.
        We can derive a useful bound for the second difference of this sequence for all $i\geq m$.
        \begin{align}
            2^{i+2}(y_i+y_{i+2}) & = 4\sum_{n=0}^{i}z_n + \sum_{n=0}^{i+2}z_n \nonumber                                                                                  \\
                                 & \leq 4\sum_{n=0}^{i}z_n + hz_{i+1}                                  &  & \text{(Equation \myeqref{eq:no_speed:lower:k_cond})}\nonumber  \\
                                 & = 4\sum_{n=0}^{i+1}z_n - (4-h)z_{i+1} \nonumber                                                                                       \\
                                 & = 2^{i+3}y_{i+1} - (4-h)z_{i+1}                                     &  & \text{(Equation \myeqref{eq:modified_y_sequence})} \nonumber   \\
                                 & \leq 2^{i+3}y_{i+1} - \left(\frac{4-h}{h}\right)\sum_{n=0}^{i+2}z_n &  & \text{(Equation \myeqref{eq:no_speed:lower:k_cond})} \nonumber \\
                                 & = 2^{i+3}y_{i+1} - 2^{i+2}\left(\frac{4-h}{h}\right)y_{i+2} .       &  & \text{(Equation \myeqref{eq:modified_y_sequence})} \nonumber
        \end{align}
        Let $\gamma = \frac{4-h}{h} > 0$. Observe that
        \begin{align}
            2^{i+2}(y_i+y_{i+2}) \;\;\; & \leq\;\;\; 2^{i+3}y_{i+1} - 2^{i+2}\gamma \cdot y_{i+2}, \nonumber         \\
            y_{i+2}-2y_{i+1}+y_i \;\;\; & \leq\;\;\;  - \gamma \cdot y_{i+2}. \mylabel{eq:no_speed:lower:final_concav}
        \end{align}
        Hence, for any $i\geq \max(m,2)$, the decrease in slope that occurs between the indices $i$ and $i+2$ in $\left(y_i\right)_{i=0}^{\infty}$ has a magnitude no less than $\gamma \cdot y_{i+2}>0$.

        A sequence that is both positive and concave down can never decrease.
        If it ever did, then the slope could never increase again and the sequence would eventually cross the $y$ axis, contradicting its positivity.
        This means that the second difference of $\left(y_i\right)_{i=0}^{\infty}$ must tend to $0$, lest the slope decrease to arbitrarily low values.
        According to Equation \myeqref{eq:no_speed:lower:final_concav}, this can only happen if $y_i$ itself tends to $0$.
        However, being a positive sequence, it must decrease at some point in order to tend to $0$.
        Hence, the existence of a sequence $\left(y_i\right)_{i=0}^{\infty}$ satisfying the above conditions is contradictory, proving Lemma~\myref{lem:no_speed:lower:k_condition}. ~\qed
    \end{proof}

    Lemma~\myref{lem:no_speed:lower:k_condition} shows that Equation \myeqref{eq:no_speed:lower:main_cond} is contradictory:
    take $z_i=w_i$, $m=\max(c_1, c_2)$, and $h=k_2<4$.
    Note that $w_i>w_{i-2}\geq 0$ for all $i\geq 2$.
    The only unfounded assumption made in the derivation of Equation \myeqref{eq:no_speed:lower:main_cond} was that $CR_{\mathcal{A}}(u)\in O\left(u^k\right)$ for $k<4$.
    Therefore, this assumption is false.
    This completes the proof of Theorem \myref{thm:no_speed:lower}. ~\qed
\end{proof}

\subsection{Upper Bound}\mylabel{sec:no_speed:upper}
In this subsection, we present an algorithm
and prove an upper bound on its performance in Theorem \myref{thm:nospeed_upbound}.
Algorithm~\myref{alg:no_speed:upper:sqrt_for_upbound} follows a zigzag strategy $\mathcal{A}$ given by $u_i = 2^{3\cdot 2^i\sqrt{i+1}-1}$.

\begin{algorithm}[H]
    \caption{\mylabel{alg:no_speed:upper:sqrt_for_upbound}(Search with unknown speed and known initial distance)}
    \begin{algorithmic}[1]
        \State \textbf{input:} initial target distance $d$
        \State terminate at any point if collocated with target
        \State $t \gets 0$
        \State \textbf{for} $i \gets 0,1,2, ..., \infty$ \textbf{do} \Indent
        \State $u_i \gets 2^{3\cdot 2^i\sqrt{i+1}-1}$
        \State $v_i \gets 1 - u_i^{-1}$
        \State $x_i \gets u_i\left(d + v_it\right)$
        \State move to position $(-1)^ix_i$ and back to the origin
        \State $t \gets t + 2x_i$ \EndIndent
        \State \textbf{endfor}
    \end{algorithmic}
\end{algorithm}

\begin{theorem} \mylabel{thm:nospeed_upbound}
    In the case where initial distance $d$ is known but speed $v$ is unknown, Algorithm~\myref{alg:no_speed:upper:sqrt_for_upbound} has the following bounds on its performance:
    \begin{align}
        \forall \; 1\leq u\leq 4 &  & CR_{\mathcal{A}}(u) & = 9 \mylabel{eq:no_speed:upper:upboundthm_early}                                   \\
        \forall \; u > 4         &  & CR_{\mathcal{A}}(u) & \leq 56.18u^{4-(\log_2 \log_2 u)^{-2}} \mylabel{eq:no_speed:upper:upboundthm_late}
    \end{align}
\end{theorem}
\begin{mainonly}
The full proof of Theorem \myref{thm:nospeed_upbound} can be found in the full version of the paper~\cite{arxiv_version}, though an outline is presented here. The proof relies on the six lemmas \myref{lem:no_speed:lower:cru_value} through \myref{lem:no_speed:upper:diff_bounds}. Lemma \myref{lem:no_speed:upper:only_key_pts} is of particular importance.
\end{mainonly}
\begin{appendixonly}
Before proceeding with the proof of Theorem \ref{thm:nospeed_upbound}, we prove the six lemmas \ref{lem:no_speed:lower:cru_value} through \ref{lem:no_speed:upper:diff_bounds}.
Lemma~\ref{lem:no_speed:upper:diff_deriv_positivity} only helps to prove Lemma~\ref{lem:no_speed:upper:diff_bounds} and does not appear in the rest of the proof.
\end{appendixonly}
\begin{lemma} \mylabel{lem:no_speed:lower:cru_value}
    For any zigzag strategy $\mathcal{A}$, $CR^\sigma_{\mathcal{A}}(u) = 1+\frac{2}{d}s_{k-1}$, where $k$ is the round on which the target with evasiveness $u$ on side $\sigma$ is caught.
\end{lemma}
\begin{appendixonly}
    \begin{proof}
        The robot begins round $k$ at the origin.
        A total of $2s_{k-1}$ time has been spent on previous rounds, giving the target an additional $2vs_{k-1}$ distance head start.
        Therefore, the time required to catch the target is:
        \begin{align*}
            T^\sigma_{\mathcal{A}}(u)  & = 2s_{k-1} + u\left(d + 2vs_{k-1}\right)
            = 2s_{k-1} + ud + 2(uv)s_{k-1}                                        \\
                                       & = 2s_{k-1} + ud + 2(u-1)s_{k-1}
            = ud + 2us_{k-1}                                                      \\
            CR^\sigma_{\mathcal{A}}(u) & = \frac{1}{ud}T^\sigma_{\mathcal{A}}(u)
            = 1 + \frac{2}{d}s_{k-1}
        \end{align*}
        This proves Lemma~\myref{lem:no_speed:lower:cru_value}. ~\qed
    \end{proof}
\end{appendixonly}

\begin{lemma} \mylabel{lem:no_speed:upper:only_key_pts}
    Consider a zigzag strategy $\mathcal{A}$ whose defining sequence of evasiveness values $\left(u_i\right)_{i=0}^{\infty}$ is strictly increasing.
    For a round number $k\geq 0$, if a function $F(u)$ is increasing for all $u \geq u_k$,
    then statement \myeqref{eq:key_pts_eq_key} implies statement \myeqref{eq:key_pts_eq_all}.
    \begin{align}
        1+\frac{2}{d}s_{i+1} \leq F(u_i) &  & \forall i\geq k \mylabel{eq:key_pts_eq_key} ,  \\
        CR_{\mathcal{A}}(u) \leq F(u)    &  & \forall u>u_{k}  \mylabel{eq:key_pts_eq_all} .
    \end{align}
\end{lemma}
\begin{proof}
    Consider any round $i$ with parity $\sigma$.
    We know that all evasiveness values $u_i < u \leq u_{i+2}$ on side $\sigma$ are caught on round $i+2$.
    By Lemma~\myref{lem:no_speed:lower:cru_value}, $CR^\sigma_{\mathcal{A}}(u)=1+\frac{2}{d}s_{i+1}$ for all $u_i < u \leq u_{i+2}$.
    Consequently, if $F(u_i) \geq 1+\frac{2}{d}s_{i+1}$ and $F(u_i)$ is increasing within the interval $[u_i, u_{i+2}]$, then $F(u) \geq CR^\sigma_{\mathcal{A}}(u)$ for all $u_i < u \leq u_{i+2}$.
    Let $\sigma_k$ denote the parity of $k$.
    Statement \myeqref{eq:key_pts_eq_key} allows us to apply our prior reasoning to all rounds $i\geq k$, collectively yielding the following:
    \begin{align*}
        F(u) & \geq CR^{\sigma_k}_{\mathcal{A}}(u)   &  & \forall u > u_k ,     \\
        F(u) & \geq CR^{1-\sigma_k}_{\mathcal{A}}(u) &  & \forall u > u_{k+1} .
    \end{align*}
    The latter statement can be extended to all $u>u_k$ as follows.
    Statement \myeqref{eq:key_pts_eq_key} yields the following in conjunction with the fact that $s_{k+1} \geq s_k$:
    \begin{align*}
        F(u) \geq 1+\frac{2}{d}s_{k+1} \geq 1+\frac{2}{d}s_{k} &  & \forall u_k < u \leq u_{k+2} .
    \end{align*}
    Recall the stipulation that $\left(u_i\right)_{i=0}^{\infty}$ is strictly increasing.
    We can reduce the applicable range of the statement above to $u_k < u \leq u_{k+1}$, which is a subset of the range $u_{k-1} < u \leq u_{k+1}$ on which $1+\frac{2}{d}s_{k} = CR^{1-\sigma_k}_{\mathcal{A}}(u)$.
    In the event that $k=0$, this range is instead $1 \leq u \leq u_{k+1}$, which still contains all $u_k < u \leq u_{k+1}$.
    From this, we can conclude that $F(u)\geq CR^{1-\sigma_k}_{\mathcal{A}}(u)$ for all $u_k < u \leq u_{k+1}$.
    \begin{align*}
        F(u) & \geq CR^{\sigma_k}_{\mathcal{A}}(u)\;,\;F(u) \geq CR^{1-\sigma_k}_{\mathcal{A}}(u)                         &  & \forall u > u_k , \\
        F(u) & \geq \max\left(CR^{\sigma_k}_{\mathcal{A}}(u),CR^{1-\sigma_k}_{\mathcal{A}}(u)\right)= CR_{\mathcal{A}}(u) &  & \forall u > u_k .
    \end{align*}
    This proves Lemma~\myref{lem:no_speed:upper:only_key_pts}. ~\qed
\end{proof}

\begin{lemma}\mylabel{lem:sum_leq_prod}
    Given a zigzag strategy, for any $i\geq 0$,
    $    s_i \leq 2^{i+1}\cdot d\prod_{j=0}^{i}u_j .
    $
\end{lemma}
\begin{appendixonly}
    \begin{proof}
        Lemma~\myref{lem:si_recurrence} states that $s_i = u_id+\left(2u_i-1\right)s_{i-1}$.
        One can easily verify that this recurrence unwinds as follows:
        \begin{equation*}
            s_i = d\sum_{j=0}^{i}u_j\prod_{k=j+1}^{i}(2u_k-1) .
        \end{equation*}
        For the purpose of finding an upper bound, let us ignore the $-1$ and simplify.
        \begin{equation*}
            s_i \leq d\sum_{j=0}^{i}\left(2^{i-j}\prod_{k=j}^{i}u_k\right) .
        \end{equation*}
        Note that since $j\geq 0$ and $u_k \geq 1$ for any $k$,
        $ \prod_{k=j}^{i}u_k \leq \prod_{k=0}^{i}u_k .
        $
        Therefore,
        \begin{align*}
            s_i & \leq d\left(\prod_{k=0}^{i}u_k\right)\sum_{j=0}^{i}2^{i-j}
            = d\left(\prod_{k=0}^{i}u_k\right)\left(2^{i+1}-1\right)
            \leq 2^{i+1} \cdot d\prod_{k=0}^{i}u_k .
        \end{align*}
        This proves Lemma~\myref{lem:sum_leq_prod}. ~\qed
    \end{proof}
\end{appendixonly}

The following notation will be used in the upcoming Lemmas.
For $n\geq 0$, let $f^{<n>}(x)$ denote the ``$n$-th difference'' of a function $f$ at $x$ as
$        f^{<0>}(x) = f(x),~
    f^{<n>}(x) = f^{<n-1>}(x) - f^{<n-1>}(x-1).
$
There are obvious parallels between $f^{<n>}(x)$ and the $n$-th derivative, $f^{(n)}(x)$.
Note that difference operations tend to alter the domain of $f$ more significantly than differentiation.
The useful properties listed below apply only where the result is defined.
Just like with regular derivatives,
$    \left(f^{<a>}\right)^{<b>}(x) = f^{<a+b>}(x) .
$
Additionally, by the linearity of differentiation,
\begin{align*}
    \left(f^{<1>}\right)^{(1)}(x) & = \frac{\mathrm{d}}{\mathrm{d}x}\left(f(x) - f(x-1)\right)
    = f^{(1)}(x) - f^{(1)}(x-1)
    = \left(f^{(1)}\right)^{<1>}(x) .
\end{align*}
In conjunction with the properties of regular derivatives, the above two properties reveal that when a function $f(x)$ is differentiated $a$ times and has $b$ first difference operations applied, the order in which all $a+b$ operations are applied does not matter.
Let us write the result of such a procedure as $f^{(a)<b>}(x)$ or $f^{<b>(a)}(x)$.

\begin{lemma} \mylabel{lem:sum_decomp}
    Consider two arbitrary functions $f(x)$ and $g(x)$ defined on integers $x\geq 0$.
    Let $G(x)$ be shorthand for $\sum_{j=0}^{x}g(j)$ where $G(0)=g(0)$ and $G(-1)=0$.
    For any constants $n\geq 0$ and $m\geq 0$,
    \begin{equation}
        \sum_{j=0}^{n}g(j)f^{<m>}(j) = G(n)f^{<m>}(n) - \sum_{j=0}^{n-1}G(j)f^{<m+1>}(j+1) . \nonumber
    \end{equation}
\end{lemma}
\begin{appendixonly}
    \begin{proof}
        The following sum transformations are valid for all $n\geq 0$.
        \begin{align*}
                  & \sum_{j=0}^{n}g(j)f^{<m>}(j)                                                               \\
            =\;\; & G(0)f^{<m>}(0) + \sum_{j=1}^{n}g(j)f^{<m>}(j)                                              \\
            =\;\; & \sum_{j=0}^{n}G(j)f^{<m>}(j) - \sum_{j=1}^{n}G(j)f^{<m>}(j) + \sum_{j=1}^{n}g(j)f^{<m>}(j) \\
            =\;\; & \sum_{j=0}^{n}G(j)f^{<m>}(j) - \sum_{j=1}^{n}G(j-1)f^{<m>}(j)                              \\
            =\;\; & G(n)f^{<m>}(n) + \sum_{j=0}^{n-1}G(j)f^{<m>}(j) - \sum_{j=0}^{n-1}G(j)f^{<m>}(j+1)         \\
            =\;\; & G(n)f^{<m>}(n) - \sum_{j=0}^{n-1}G(j)f^{<m+1>}(j+1) .
        \end{align*}
        This proves Lemma~\myref{lem:sum_decomp}. ~\qed
    \end{proof}
\end{appendixonly}

\begin{lemma} \mylabel{lem:no_speed:upper:diff_deriv_positivity}
    For the function $f(x)=\sqrt{x}$ and for any integers $k\geq 1$ and $0\leq m\leq k$, $(-1)^{k+1}f^{(k-m)<m>}(x)$ is non-negative for all $x>m$.
\end{lemma}
\begin{appendixonly}
    \begin{proof}
        Consider some integer $k\geq 1$.
        With basic calculus, one can show that $(-1)^{k+1}f^{(k)}(x)$ is non-negative within its domain of $x>0$.
        This proves the case where $m=0$.
        The remaining cases follow using induction.
        For the sake of clarity, let $g(x) = (-1)^{k+1}f(x)$.
        Suppose that for some $0\leq m \leq k-1$, $g^{(k-m)<m>}(x)$ is non-negative for all $x>m$.
        This means that $g^{(k-m-1)<m>}(x)$ is increasing within the interval $x>m$.
        Hence, for all $x>m+1$,
        \begin{align*}
            g^{(k-m-1)<m>}(x)                     & \geq g^{(k-m-1)<m>}(x-1) \\
            g^{(k-m-1)<m>}(x)-g^{(k-m-1)<m>}(x-1) & \geq 0                   \\
            g^{(k-m-1)<m+1>}(x)                   & \geq 0
        \end{align*}
        By induction, $g^{(k-m)<m>}(x) = (-1)^{k+1}f^{(k-m)<m>}(x)$ is non-negative for all $x>m$ for any $0\leq m \leq k$, proving Lemma~\myref{lem:no_speed:upper:diff_deriv_positivity}. ~\qed
    \end{proof}
\end{appendixonly}

\begin{lemma} \mylabel{lem:no_speed:upper:diff_bounds}
    For the function $f(x) = \sqrt{x}$ specifically, $f^{<k>}(x)$ is subject to the following upper and lower bounds for any integer $k\geq 1$ for all real $x>k$:
    \begin{equation*}
        (-1)^{k+1}f^{(k)}(x-\frac{1}{2}k) \;\leq\; (-1)^{k+1}f^{<k>}(x) \;\leq\; (-1)^{k+1}f^{(k)}(x-k) .
    \end{equation*}
    Instances of $(-1)^{k+1}$ can be removed by appropriately alternating the direction of the inequality.
\end{lemma}
\begin{appendixonly}
    \begin{proof}
        Let the following statement be labelled $S(m)$ for some integer $0\leq m\leq k$:
        for any integer $k\geq 1$ and all real $x>k$,
        \begin{equation*}
            (-1)^{k+1}f^{(m)<k-m>}(x-\frac{1}{2}m) \;\leq\; (-1)^{k+1}f^{<k>}(x) \;\leq\; (-1)^{k+1}f^{(m)<k-m>}(x-m) .
        \end{equation*}
        We will prove statement $S(m)$ for all $0\leq m\leq k$ in order of increasing $m$ using induction.
        The final case, $S(k)$, is equivalent to Lemma~\myref{lem:no_speed:upper:diff_bounds}.
        Our base case, $S(0)$, is a trivial statement as all three members of the inequality are equal.

        As an inductive hypothesis, assume that $S(m)$ holds for some $0\leq m\leq k-1$.
        Consider the function $g(x) = (-1)^{k+1}f^{(m)<k-m-1>}(x)$ where it is guaranteed that $k-m-1\geq 0$.
        $g(x)$ is at least defined for all $x>k-m-1$.
        According to Lemma~\myref{lem:no_speed:upper:diff_deriv_positivity}, $g(x)$ is concave down, meaning that
        \begin{align*}
            g(x)                        & \leq g(x-1) + g'(x-1)                  &  & x>k-m \\
            g^{<1>}(x)                  & \leq g'(x-1)                           &  & x>k-m \\
            g^{<1>}(x-m)                & \leq g'(x-m-1)                         &  & x>k   \\
            (-1)^{k+1}f^{(m)<k-m>}(x-m) & \leq (-1)^{k+1}f^{(m+1)<k-m-1>}(x-m-1) &  & x>k
        \end{align*}
        In conjunction with the inductive hypothesis, this proves the upper bound portion of statement $S(m+1)$.
        Next, we consider the lower bound portion.
        $g^{<1>}(x)$ can be interpreted as the area under $g'(x)$ within the interval from $x-1$ to $x$.
        According to Lemma~\myref{lem:no_speed:upper:diff_deriv_positivity}, $g'(x)$ is non-negative and concave up.
        As a result, $g'(x)$ is always at or above the line tangent to $g'(x)$ at the point $\left(x-\frac{1}{2},g'(x-\frac{1}{2})\right)$.
        Within the interval from $x-1$ to $x$, the signed area under that tangent line (regardless of its slope) is equal to $g'(x-\frac{1}{2})$.
        Since the tangent line is always at or below $g'(x)$, we can conclude that $g^{<1>}(x) \geq g'(x-\frac{1}{2})$, for $x>k-m$, and that $g^{<1>}(x-\frac{1}{2}m) \geq g'(x-\frac{1}{2}(m+1))$, for $k-\frac{1}{2}m < k < x$.     
        Therefore,
        \begin{align*}
            (-1)^{k+1}f^{(m)<k-m>}(x-\frac{1}{2}m) & \geq (-1)^{k+1}f^{(m+1)<k-m-1>}(x-\frac{1}{2}(m+1))
        \end{align*}
        for $x>k$. Together with the inductive hypothesis, this proves the lower bound portion of statement $S(m+1)$.
        By induction, $S(k)$ is true, which proves Lemma~\myref{lem:no_speed:upper:diff_bounds}. ~\qed
    \end{proof}
\end{appendixonly}

\begin{mainonly}
    \begin{proof} (Theorem \myref{thm:nospeed_upbound} - sketch)
        We begin by imagining an upper bound of the form $CR_{\mathcal{A}}(u) \leq F(u) = cu^{4-h(u)}$. Lemma \myref{lem:no_speed:upper:only_key_pts} reveals that this condition holds for all $u>u_1$ if $F(u)$ is increasing on that range and $1+\frac2ds_{i+1} \leq F(u_i)$ for all $i\geq 1$. From $F(u_i)$, we isolate $h(u_i)$ to obtain the following:
        \begin{equation*}
            h(u_i) \leq \frac{4\log_2(u_i) - \log_2\left(1+\frac{2}{d}s_{i+1}\right) + \log_2 c}{\log_2(u_i)} \qquad \forall i\geq 1.
        \end{equation*}
        We continue to transform this condition so that it is more strict but easier to work with. Lemma \myref{lem:sum_leq_prod} converts $s_{i+1}$ into a product, which becomes a sum under the logarithm. We apply Lemma \myref{lem:sum_decomp} to the resulting sum three times, which leaves us with many terms, a few of which cleanly cancel the $4\log_2(u_i)$ term. We then reduce all but one of the remaining terms to constants with the help of Lemma \myref{lem:no_speed:upper:diff_bounds}; $c=56.18$ is large enough that the $\log_2 c$ term cancels those constants. The remaining term is simplified using Lemma \myref{lem:no_speed:upper:diff_bounds}. It can then cancel out the denominator, leaving $(i+1)^{-2}$ on the right-hand side, which is larger than $h(u_i) = \left(\log_2 \log_2 u_i\right)^{-2}$. We note that for such an $h(u)$, the resulting function $F(u)$ is increasing where needed. We conclude that $F(u) = 56.18u^{4-\left(\log_2 \log_2 u\right)^{-2}}$ is a valid upper bound on $CR_{\mathcal{A}}(u)$ for all $u>u_1$.

        In the case where $u\leq u_1$, we calculate $CR_{\mathcal{A}}(u)$ precisely using Lemma \myref{lem:no_speed:lower:cru_value}. As a consequence of Lemma \myref{lem:no_speed:lower:cru_value}, $CR_{\mathcal{A}}(u)$ only has two values in this range. We find that $F(u)$ remains a valid upper bound for all $4<u\leq u_1$. Meanwhile, $CR_{\mathcal{A}}(u)=9$ for all $u\leq 4$. This concludes the proof sketch for Theorem \myref{thm:nospeed_upbound}.~\qed
    \end{proof}
\end{mainonly}

\begin{appendixonly}
\begin{proof}[Theorem~\myref{thm:nospeed_upbound}]
    With all of the important lemmas proven, let us proceed with the proof of Theorem \myref{thm:nospeed_upbound}.
    Let us formulate our hypothetical upper bound in terms of some function $h(u)$ as follows:
    \begin{align}
        CR_{\mathcal{A}}(u) \leq F(u) = cu^{4-h(u)} &  & \forall u > u_{crit} . \mylabel{eq:basic_bigo_h}
    \end{align}
    The bulk of the proof will show that this inequality is satisfied with $h(u)=(\log_2 \log_2 u)^{-2}$, $c = 56.18$, and $u_{crit} = u_1 \approx 179.18$.
    Cases where $u\leq u_1$ will be addressed at the end of the proof.

    To begin, Lemma~\myref{lem:no_speed:upper:only_key_pts} reveals that Equation \myeqref{eq:basic_bigo_h} can be satisfied for all $u > u_1$ with the following:
    \begin{align}
        1+\frac{2}{d}s_{i+1} \leq cu_i^{4-h(u_i)} &  & \forall i\geq 1 . \mylabel{eq:key_pt_bigo}
    \end{align}
    The values of $c$ and $h(u)$ will be gradually derived.
    $h(u_i)$ can be isolated as follows:
    \begin{equation}
        h(u_i) \leq \frac{4\log_2(u_i) - \log_2\left(1+\frac{2}{d}s_{i+1}\right) + \log_2 c}{\log_2(u_i)} \qquad \forall i\geq 0 . \mylabel{eq:key_pt_bigo_iso}
    \end{equation}
    Our approach from this point forward will be to repeatedly replace the right hand side of Equation \myeqref{eq:key_pt_bigo_iso} with expression that are lesser or equal in value. This ensures that the $h(u)$ we derive satisfies Equation \myeqref{eq:key_pt_bigo_iso}.
    Let the following portion of Equation \myeqref{eq:key_pt_bigo_iso} be dubbed $H(i)$:
    \begin{equation}
        H(i) = 4\log_2(u_i)-\log_2\left(1+\frac{2}{d}s_{i+1}\right) . \mylabel{eq:important_sum}
    \end{equation}
    We will derive a lower bound for $H(i)$ in which many parts have cancelled one another, then substitute it into Equation \myeqref{eq:key_pt_bigo_iso}.
    First, we use Lemma~\myref{lem:sum_leq_prod} to obtain an upper bound on $s_{i+1}$.
    \begin{equation*}
        H(i) \geq 4\log_2(u_i)-\log_2\left(1+8\cdot 2^i\prod_{j=0}^{i+1}u_j\right) .
    \end{equation*}
    Since $u_j\geq 1$ for all $j\geq 0$,
    \begin{align*}
        H(i) & \geq 4\log_2(u_i)-\log_2\left(9\cdot 2^i\prod_{j=0}^{i+1}u_j\right) \\
             & = 4\log_2(u_i)-\log_2 9 - i - \sum_{j=0}^{i+1}\log_2(u_j) .
    \end{align*}
    Substituting in the concrete value of $\log_2(u_j)$ yields
    \begin{align}
        H(i) & \geq \notag 4\left(3\cdot 2^i\sqrt{i+1}-1\right)-\log_2 9 - i - \sum_{j=0}^{i+1}\left(3\cdot 2^j\sqrt{j+1}-1\right) \\
             & = \notag 3\cdot 2^{i+2}\sqrt{i+1}-4-\log_2 9 - i + (i+2) - 3\sum_{j=0}^{i+1}2^j\sqrt{j+1}                           \\
             & = \mylabel{eq:precancel_sum} 3\cdot 2^{i+2}\sqrt{i+1}-2-\log_2 9 - 3\sum_{j=0}^{i+1}2^j\sqrt{j+1} .
    \end{align}
    Next, we focus on rewriting the sum.

    We will apply Lemma~\myref{lem:sum_decomp} three times to Equation \myeqref{eq:precancel_sum}.
    Let $f(n)=\sqrt{n+1}$.
    For integers $n\geq 0$, let $g_0(n)=2^n$ and $g_k(n)=\sum_{j=0}^{n}g_{k-1}(j)$ for $1\leq k\leq 3$.
    Given this, one can verify the following:
    \begin{align*}
        g_1(n) & = 2^{n+1}-1,~
        g_2(n) = 2^{n+2}-n-3,~
        g_3(n) = 2^{n+3}-\frac{1}{2}n^2-\frac{7}{2}n-7.
    \end{align*}
    One can also verify that $g_k(n)$ is non-negative for any non-negative integers $n$ and $k$.
    To begin, Equation \myeqref{eq:precancel_sum} can be rewritten as follows:
    \begin{equation*}
        H(i) \geq 3\cdot 2^{i+2}f(i)-2-\log_2 9 - 3\sum_{j=0}^{i+1}g_0(j)f(j) .
    \end{equation*}
    Given that $i\geq 1$, applying Lemma~\myref{lem:sum_decomp} three times to the above equation yields the following:
    \begin{align*}
        H(i) \geq\; & 3\cdot 2^{i+2}f(i) -2-\log_2 9
        - 3g_1(i+1)f(i+1)
        + 3g_2(i)f^{<1>}(i+1)                        \\
                    & - 3g_3(i-1)f^{<2>}(i+1)
        + 3\sum_{j=0}^{i-2}g_3(j)f^{<3>}(j+3) .
    \end{align*}
    Let us rewrite this in a more coherent fashion and substitute in concrete values for $f$ and $g$ where needed.
    Note that the concrete values for $g$ can be broken up into two terms: a positive and a negative one.

    \begin{align*}
        H(i) \geq & -2-\log_2 9                  &  & + 3\cdot 2^{i+2}f(i)                                      &  & \text{(Terms 1,2)} \\
                  & - 3\cdot 2^{i+2}f(i+1)       &  & +3f(i+1)                                                  &  & \text{(Terms 3,4)} \\
                  & + 3\cdot 2^{i+2}f^{<1>}(i+1) &  & -3(i+3)f^{<1>}(i+1)                                       &  & \text{(Terms 5,6)} \\
                  & - 3\cdot 2^{i+2}f^{<2>}(i+1) &  & + 3\left(\frac{1}{2}i^2+\frac{5}{2}i+4\right)f^{<2>}(i+1) &  & \text{(Terms 7,8)} \\
                  &                              &  & + 3\sum_{j=0}^{i-2}g_3(j)f^{<3>}(j+3) .                   &  & \text{(Term 9)}
    \end{align*}
    We will examine each of the above Terms 1-9 individually.
    By definition, $f^{<1>}(i+1) - f(i+1) + f(i) = 0$.
    As a result, terms 2, 3, and 5 sum to $0$.
    Term 9 is non-negative since both $g_3(j)$ and $f^{<3>}(j+3)$ are non-negative for all $j\geq 0$.
    This means that term 9 can be discarded in the pursuit of a simpler lower bound on $H(i)$.
    Let $\Phi(i)$ represent the sum of terms 4, 6, and 8:
    \begin{equation*}
        \Phi(i) = 3f(i+1) - 3(i+3)f^{<1>}(i+1) + 3\left(\frac{1}{2}i^2+\frac{5}{2}i+4\right)f^{<2>}(i+1) .
    \end{equation*}
    We will show that across all integers $i\geq 1$, $\Phi(i)$ has a minimum value of $\Phi(1)$, which equals $12\sqrt{3}-30\sqrt{2}+21 \approx -0.642$.
    Lemma~\myref{lem:no_speed:upper:diff_bounds} provides the following lower bound for $\Phi(i)$ which is valid for $i\geq 1$:
    \begin{align*}
        \Phi(i) & \geq 3\sqrt{i+2} - 3(i+3)\cdot \frac{1}{2\sqrt{i+1}} - 3\left(\frac{1}{2}i^2+\frac{5}{2}i+4\right)\cdot \frac{1}{4}i^{-\frac{3}{2}} \\
                & \geq 3\sqrt{i} - 3(i+3)\cdot \frac{1}{2\sqrt{i}} - 3\left(\frac{1}{2}i^2+\frac{5}{2}i+4\right)\cdot \frac{1}{4}i^{-\frac{3}{2}}     \\
                & = \frac{9}{8}\sqrt{i} -\frac{51}{8}i^{-\frac{1}{2}} - 3i^{-\frac{3}{2}} .
    \end{align*}
    This lower bound is clearly increasing since all constituent terms are increasing.
    At $i=5$, the binding function reaches a value of $-\frac{27\sqrt{5}}{100} \approx -0.604$ which exceeds $\Phi(1)$.
    This guarantees that $\Phi(i)$ exceeds $\Phi(1)$ for all $i\geq 5$.
    From here, one can simply evaluate $\Phi(i)$ at the remaining values $i=2,3,4$ to discover that $\Phi(i)\geq -0.642$ for all integers $i\geq 1$.
    This allows us to replace terms 4,6,8 with a lower bound of $-0.642$ in $H(i)$, leaving only Term 7 (see last Inequality for $H(i)$ above) and a few constants:
    $    H(i) \geq -2-\log_2 9-0.642 - 3\cdot 2^{i+2}f^{<2>}(i+1).
    $
    Lemma~\myref{lem:no_speed:upper:diff_bounds} allows us to place the needed upper bound on $f^{<2>}(i+1)$:
    \begin{align*}
        f^{<2>}(i+1)                 & \leq -\frac{1}{4}(i+1)^{-\frac{3}{2}}                 \\
        - 3\cdot 2^{i+2}f^{<2>}(i+1) & \geq 3\cdot 2^i(i+1)^{-\frac{3}{2}}                   \\
        H(i)                         & \geq 3\cdot 2^i(i+1)^{-\frac{3}{2}} -2-\log_2 9-0.642
    \end{align*}

    We now have a satisfactory lower bound on $H(i)$ that can be substituted into Equation \myeqref{eq:key_pt_bigo_iso}:
    \begin{equation*}
        h(u_i) \leq \frac{-2-\log_2 9 - 0.642 + 3\cdot 2^{i}(i+1)^{-\frac{3}{2}} + \log_2 c}{\log_2(u_i)} \qquad \forall i\geq 1 .
    \end{equation*}
    Selecting $c=56.18$ is sufficient to cancel out all other constant terms.
    Substituting the concrete value of $\log_2(u_i)$ yields
    \begin{align*}
         & \frac{3\cdot 2^{i}(i+1)^{-\frac{3}{2}}}{3\cdot 2^i\sqrt{i+1}-1}
        \geq \frac{3\cdot 2^{i}(i+1)^{-\frac{3}{2}}}{3\cdot 2^i\sqrt{i+1}}
        = (i+1)^{-2} .
    \end{align*}
    For $i \geq 1$,
    \begin{align*}
        3\cdot 2^i\sqrt{i+1}-1               & \geq 2^{i+1}    \\
        \log_2(u_i)                          & \geq 2^{i+1}    \\
        \log_2 \log_2(u_i)                   & \geq i+1        \\
        \left(\log_2 \log_2(u_i)\right)^{-2} & \leq (i+1)^{-2}
    \end{align*}

    In summary, $h(u) = \left(\log_2 \log_2(u)\right)^{-2}$ and $c=56.18$ are sure to satisfy the following statement, which is a key criteria for Lemma~\myref{lem:no_speed:upper:only_key_pts}:
    \begin{align*}
        1+\frac{2}{d}s_{i+1} \leq cu_i^{4-h(u_i)} &  & \forall i\geq 1 .
    \end{align*}
    $u_1$ has the exact value $2^{6\sqrt{2}-1} \approx 179.18$.
    $F(u)=cu^{4-h(u)}$ is definitely increasing for all $u>2$ because $h(u)$ is positive and decreasing while $u$ is greater than $1$ and increasing.
    This justifies applying Lemma~\myref{lem:no_speed:upper:only_key_pts} to conclude the following:
    \begin{align*}
        CR_{\mathcal{A}}(u) \leq 56.18u^{4-\left(\log_2 \log_2(u)\right)^{-2}} &  & \forall u > u_1 .
    \end{align*}
    This is almost sufficient to prove statement \myeqref{eq:no_speed:upper:upboundthm_late} of Theorem \myref{thm:nospeed_upbound};
    it covers all $u > u_1$.

    Earlier values of $CR_{\mathcal{A}}(u)$ can be calculated manually using Lemma~\myref{lem:si_recurrence}:
    $
        s_i = u_id + (2u_i-1)s_{i-1}
    $, meaning that
    \begin{align*}
        u_0    & = 4,~
        u_1 = 2^{6\sqrt{2}-1} \approx 179.18 ,~
        u_2 = 2^{12\sqrt{3}-1} \approx 903152.25 \\
        s_{-1} & = 0 ,~
        s_{0} = 4d ,~
        s_{1} = \left(\frac{9}{2}\cdot 2^{6\sqrt{2}}-4\right)d \approx 1608.64d .
    \end{align*}
    Substituting these values into Lemma~\myref{lem:no_speed:lower:cru_value} yields the following:
    \begin{align*}
        CR^0_{\mathcal{A}}(u) & = 1+\frac{2}{d}s_{-1}    &  & \forall \; 1\leq u\leq u_0                           \\
                              & = 1                      &  & \forall\; 1\leq u\leq 4                              \\
        CR^1_{\mathcal{A}}(u) & = 1+\frac{2}{d}s_{0}     &  & \forall \; 1\leq u\leq u_1                           \\
                              & = 9                      &  & \forall\; 1\leq u\leq 2^{6\sqrt{2}-1} \approx 179.18 \\
        CR^0_{\mathcal{A}}(u) & = 1+\frac{2}{d}s_{1}     &  & \forall \; u_0< u\leq u_2                            \\
                              & = 9\cdot 2^{6\sqrt{2}}-7 &  & \forall\; 4<u\leq 2^{12\sqrt{3}-1} \approx 903152.25 \\
                              & \approx 3218.27 .
    \end{align*}
    From this, we can conclude the following about $CR_{\mathcal{A}}(u)$:
    \begin{align*}
        CR_{\mathcal{A}}(u) & = 9             & 1\leq u \leq 4  \\
        CR_{\mathcal{A}}(u) & \approx 3218.27 & 4< u \leq u_1 .
    \end{align*}
    The latter statement finishes the proof of statement \myeqref{eq:no_speed:upper:upboundthm_late}.
    We know that $F(u)$ is increasing for all $u>2$ and that $F(4) = 3595.52 > 3218.27$, meaning that $F(u) \geq CR_{\mathcal{A}}(u)$ for all $4<u\leq u_1$.
    This concludes the proof of Theorem \myref{thm:nospeed_upbound} and all of the statements comprising it. ~\qed
\end{proof}
\end{appendixonly}

\section{Unknown Speed and Starting Distance}
\mylabel{sec:no_knowledge}

In this section, we analyze the competitive ratio of search when both the speed and starting distance of the mobile target are unknown to the searcher. Our proof of the upper bound will make use of the following lemma.

\begin{lemma} \mylabel{lem:d1_worstcase}
    $T^\sigma_{\mathcal{A}}(u,d)\leq T^\sigma_{\mathcal{A}}(ud,1)$ for $\sigma\in \{0,1\}$ and $CR_{\mathcal{A}}(u,d)\leq CR_{\mathcal{A}}(ud,1)$, for all strategies $\mathcal{A}$.
\end{lemma}
\begin{proof}
    Consider two moving targets: the first with $u_1=u, d_1=d$ and the second with $u_2=ud, d_2=1$, and respective speeds $v_1 = 1 -\frac1u$ and $v_2=1-\frac1{ud}$. The first target is not further from the origin than the second target at all times $t\geq ud$, and this is also the earliest time at which either target can be caught by the searcher. If we note that
    \begin{align*}
        t \geq ud
         & \Leftrightarrow t\left(u^{-1}-(ud)^{-1}\right) \geq d-1
        \Leftrightarrow 1 + (1-(ud)^{-1})t \geq d + (1-u^{-1})t
    \end{align*}
    then we conclude that the first target cannot be caught after the second target. This implies that $T^\sigma_{\mathcal{A}}(u,d) \leq T^\sigma_{\mathcal{A}}(ud,1)$ for $\sigma\in\{0,1\}$. It follows that
    \begin{align*}
        \frac{\max\left\{T^0_{\mathcal{A}}(u,d), T^1_{\mathcal{A}}(u,d)\right\}}{ud} & \leq\frac{\max\left\{T^0_{\mathcal{A}}(ud,1), T^1_{\mathcal{A}}(ud,1)\right\}}{ud}
    \end{align*}
    implying that $CR_{\mathcal{A}}(u,d) \leq CR_{\mathcal{A}}(ud,1)$, which proves Lemma~\myref{lem:d1_worstcase}. ~\qed
\end{proof}

Motivated by Lemma~\myref{lem:d1_worstcase}, one is led to consider running the previous search Algorithm~\myref{alg:no_speed:upper:sqrt_for_upbound} under the assumption that the target's initial distance is $1$.
\begin{algorithm}[H]
    \caption{\mylabel{alg:no_knowledge} (Search with unknown speed and initial distance)}
    \begin{algorithmic}[1]
        \State Execute Algorithm~\myref{alg:no_speed:upper:sqrt_for_upbound} as though $d=1$ (regardless of the true value of $d$)
    \end{algorithmic}
\end{algorithm}

\begin{theorem} \mylabel{thm:noknow_upbound_cr2param}
    In the case where both initial distance $d$ and speed $v$ are unknown, the competitive ratio of Algorithm~\myref{alg:no_knowledge} satisfies the following bounds:
    \begin{align}
        CR_{\mathcal{A}}(u,d) & \leq \mylabel{eq:noknow_upbound_eq12}
        \left\{
        \begin{array}{ll}
            1+\frac8d                                                           & \mbox{~if $ud \leq 4$} \\
            1+\frac{1}{d}\left(56.18(ud)^{4-(\log_2 \log_2 (ud))^{-2}}-1\right) & \mbox{~if $ud > 4$}
        \end{array}
        \right.
    \end{align}
\end{theorem}
\begin{proof}
    By the same reasoning as in Lemma~\myref{lem:no_speed:lower:cru_value}, $CR_{\mathcal{A}}(u,d) = 1+\frac{2}{d}s_{k-1}$ where $k$ is the round on which the target with evasiveness $u$ and initial distance $d$ on side $\sigma$ is caught. By extension, $s_{k-1} = \frac{d}{2}\left(CR_{\mathcal{A}}(u,d)-1\right)$. Consider two targets: the first with $u_1=u, d_1=d$ and the second with $u_2=ud, d_2=1$.
    By Lemma~\myref{lem:d1_worstcase}, $T_{\mathcal{A}}(u,d)\leq T_{\mathcal{A}}(ud,1)$, meaning that target 2 is caught no sooner than target 1. If we define $k_1,k_2$ to be the rounds on which each respective target is caught, we can conclude that $k_1\leq k_2$ and that $s_{k_1-1}\leq s_{k_2-1}$. It follows from our prior reasoning that
    $    \frac{d}{2}\left(CR_{\mathcal{A}}(u,d)-1\right) \leq\frac{1}{2}\left(CR_{\mathcal{A}}(ud,1)-1\right)
    $.
    Finally, isolating $CR_{\mathcal{A}}(u,d)$ yields
    $        CR_{\mathcal{A}}(u,d) \leq1+\frac{1}{d}\left(CR_{\mathcal{A}}(ud,1)-1\right) .
    $

    Algorithm~\myref{alg:no_knowledge} runs Algorithm~\myref{alg:no_speed:upper:sqrt_for_upbound} under the assumption that the starting distance of the mobile target is $1$, meaning that Theorem \myref{thm:nospeed_upbound} provides an upper bound on $CR_{\mathcal{A}}(u,1)$. Since we are interested in $CR_{\mathcal{A}}(ud,1)$, we replace occurrences of $u$ with $ud$. Doing so provides us with statement \myeqref{eq:noknow_upbound_eq12}, concluding the proof of Theorem \myref{thm:noknow_upbound_cr2param}. ~\qed
\end{proof}

We note that since $ud \leq \max\left\{u,d\right\}^2 = M^2$, the upper bound presented in Theorem \myref{thm:noknow_upbound_cr2param} is a strict improvement (asymptotically) over the previous best known bound, $O \left(\frac1d M^8 \log_2^2 M \log_2 \log_2 M\right)$, in~\cite{coleman2023line}. Also the lower bound in the case where $d$ is known extends trivially to the case where $d$ is unknown, yielding that no strategy $\mathcal{A}$ could satisfy $CR_{\mathcal{A}}(u,d) \in O\left(u^{4-\varepsilon}\right)$, for any constant $\varepsilon>0$.

\section{Conclusion}\mylabel{sec:conclusion}

In this paper, we considered linear search for an escaping oblivious mobile target by an autonomous mobile agent. Based on the competitive ratio, our algorithm and its analysis indicates optimality up to low order terms in the exponent for the case when the speed $0 \leq v <1$ of the mobile target is unknown to the searcher, thus answering an open problem in~\cite{coleman2023line}. We also analyzed and improved on previous results in~\cite{coleman2023line} when both $d,v$ are unknown; however, tight bounds for this case remain elusive. A most interesting (and challenging) direction for future research is group search  (in the context of linear search) by a multi-agent system of searchers with various communication behaviours and capabilities.

\bibliographystyle{abbrv}
\bibliography{refs}





\end{document}